\documentclass{article}
\usepackage[utf8]{inputenc}
\usepackage{bbold, mathtools, amsthm, setspace}
\usepackage[colorlinks,hypertexnames=true,citecolor=blue]{hyperref} 
\usepackage{natbib}
\usepackage{tikz}
\usepackage{pgfplots}
\usepackage{appendix}
\setcitestyle{authoryear,open={(},close={)}}
\newtheorem{theorem}{Theorem}
\newcommand{\V}{\mathbb{V}}
\newcommand{\E}{\mathbb{E}}
\newcommand{\R}{\mathbb{R}}

\title{A re-examination of the foundations of the\\ cost of capital for regulatory purposes}
\author{Darryl Biggar\footnote{Adjunct Associate Professor, Monash University.}}
\date{March 2023}

\begin{document}

\maketitle

\begin{abstract}
\noindent
In regulatory proceedings, few issues are more hotly debated than the cost of capital. This article formalises the theoretical foundation of cost of capital estimation for regulatory purposes. Several common regulatory practices lack a solid foundation in the theory.  For example, the common practice of estimating a single cost of capital for the regulated firm suffers from a circularity problem, especially in the context of a multi-year regulatory period. In addition, the relevant cost of debt cannot be estimated using the yield-to-maturity on a corporate bond. We suggest possible directions for reform of cost of capital practices in regulatory proceedings.
%

\end{abstract}

\section{Introduction}


In almost all public utility rate cases, the regulator is required to determine a parameter known as the cost of capital. Debates over this parameter are routinely amongst the most contentious issues in regulatory proceedings. Corporate finance theory provides some guidance over these debates, but much of that theory was developed for other purposes and is not always relevant or directly applicable to regulatory proceedings. This article takes a fresh look at cost of capital theory to explore whether we can strengthen the theoretical foundation of cost of capital for regulatory practices. We find that several practices that are common in regulatory proceedings do not have a sound foundation in the underlying theory.

In practice, in the estimation of cost of capital, heuristics, assumptions and approximations are common. In many cost-of-capital applications, especially in corporate finance, these can be useful and reasonable. However, in the context of litigious regulatory proceedings, regulators and courts need a basis on which to stand when choosing between competing theories, methodologies and approaches. This article starts from the perspective that in making regulatory decisions, courts and regulators should rely -- as far as possible -- on the underlying theory. Without a theoretical foundation on which to stand regulators and courts cannot make reasoned and rational decisions whether to prefer one approach to estimating the cost of capital over another, or even determine whether the question or dispute in from of them makes sense.

As we will see, this re-examination of the underlying theory raises questions about several common practices in cost of capital estimation for regulatory purposes. Wherever possible I propose alternative approaches which are more consistent with the underlying theory. However those alternative approaches may require access to information that courts or regulators find difficult to obtain. These issues will need to be worked through in regulatory practice. At this stage this article seeks to identify potential problems and to suggest potential solutions.

This article does not assume any particular methodology for the estimation of the cost of capital -- whether that is the Capital Asset Pricing Model (CAPM) (or its variants), Arbitrage Pricing Theory, or the Fama-French 3-factor model.\footnote{See the Wikipedia entry on Asset Pricing.} Rather, this article uses the concept of the `valuation functional' or `valuation operator' which is independent of any particular approach to estimating the cost of capital. The results set out here are therefore independent of, and more fundamental, than a particular methodology for valuing a cash-flow stream. We set out the properties of the value functional in section \ref{sec:prelim} below.

Cost of capital is important in regulatory processes for one central reason: In all cost-based regulatory processes a key goal is to set prices or allowed revenue which allow the regulated firm to `cover its costs' and no more.\footnote{\cite{joskow2007regulation}: ``The basic legal principle that governs price regulation in the U.S. is that regulated prices must be set at levels that give the regulated firm a reasonable opportunity to recover the costs of the investments it makes efficiently to meet its service obligations and no more than is necessary to do so.''} This is expressed formally in the objective that the regulated firm's cash-flow stream should be set to achieve a Net Present Value (NPV) of zero. In a context in which the recovery of investment costs is spread out over time, achieving the objective of NPV=0 requires an estimation of the relevant discount rate or cost of capital for each cash-flow. Errors in estimation of the cost of capital translate into errors in the NPV. Increased uncertainty in cost of capital estimation therefore results in either (a) a greater risk of under-investment or default by the regulated firm; or (b) the regulator increasing the size of the `buffer' in the allowed revenues to ensure that the regulated firm does not under-invest or default. In principle, improving the quality of cost of capital estimation would allow the regulator to lower prices to customers and/or reduce the risk to the regulated firm of under-investment or default.

There are many different ways to regulate to achieve NPV=0. The precise definition of the cost of capital in a regulatory process depends very strongly on the precise approach to regulation. It is therefore important to carefully formalise the statement of the regulatory process, which we do below.

Current regulatory practice in Australia and New Zealand can be summarised as follows: In the typical case of a regulatory period of five years, a single cost of capital parameter is estimated which applies to the entire five year period. Starting from the opening regulatory asset base, the cost of capital (together with estimates of opex, capex, and the choice of depreciation) is used to determine the allowed revenues for each year of the regulatory period, and the closing asset base.

The cost of capital parameter in this process is typically estimated using the yield on a long-term government bond (typically with a term of ten years), plus a premium for risk. The risk premium is estimated as a weighted average of the risk premium for equity and the risk premium for debt. The risk premium for equity is typically estimated using a version of the CAPM.
The risk premium for debt is typically estimated in one of two ways: (a) as the risk premium on a representative corporate bond of the correct term and credit rating; or (b) as a trailing average of a representative portfolio of corporate bonds over given term and credit rating. Although the details vary, this broad approach is consistent with a number of other regulatory jurisdictions.

Within this framework we would like to give formal answers to the following questions:
\begin{itemize}
    \item Should the cost of capital for a regulated firm vary with the level of the regulatory asset base or the level of the allowed revenues, holding all else constant? Does this give rise to a problem of circularity?
    \item In the context of a multi-year regulatory period what is the right `term' of the cost of capital? Is the right term a `one year' rate, a term equal to the length of the regulatory period, a longer term, or something else entirely?
    
    \item If we estimate the cost of capital for the regulated firm as a weighted average of the cost of equity and the cost of debt, what is the right approach to estimating the cost of debt? Can we estimate the cost of debt by observing the appropriate yield-to-maturity on a corporate bond of the right term? If so, what is the right term?

    \item Is it possible to express the single cost-of-capital for a multi-year regulatory period as a weighted average of the cost of equity and the cost of debt?
\end{itemize}

The key conclusions of this article may be summarised as follows:
\begin{itemize}

\item In common regulatory practice the regulator estimates a single cost of capital for the combined cash-flow of the regulated firm (the sum of the allowed cash-flow of the firm and the closing regulatory asset base). But this cost of capital depends, itself, on the level of the allowed cash-flow of the firm. This gives rise to a problem of circularity: A choice of the cost of capital affects the allowed cash-flow, which affects the cost of capital. At a minimum, this complicates the task of estimating the correct cost of capital for the firm. This problem can be resolved by estimating a separate cost of capital for the one-period cash-flow and for the closing asset base of the regulated firm.

\item Common regulatory practice sets a single aggregate cost of capital parameter for the entire regulatory period. This single cost of capital parameter is a complex polynomial of the individual component costs-of-capital and, because it depends on the expected values of the individual cash-flows, it once again suffers from the problem of circularity.

In the more general case where the regulatory period lasts $T$ years, there are potentially up to $T+1$ cost-of-capital-related-parameters that the regulator must estimate in regulatory proceedings -- one for each of the single-period cash-flows during the regulatory period and one for the closing asset base.


\item In common regulatory practice, the cost of capital for the firm as a whole is divided into a separate cost of equity and a cost of debt. However, we demonstrate that the cost of capital(s) for the debt payment stream typically cannot be observed directly from observations of the current yield-to-maturity on bonds traded in the market. In particular, neither the `on the day' approach (that is, the yield to maturity on a zero-coupon bond with a term equal to the length of the regulatory period) nor the `trailing average' approach (that is, the weighted average yield to maturity on the portfolio of bonds sold by the regulated firm) yields the correct cost of capital parameter for the payments to debtholders.


\item In the context of a multi-year regulatory period, the cost of capital parameter cannot be expressed as a weighted average of a cost of capital for equity and a cost of capital for debt. This calls into question the value of estimating a Weighted-Average Cost of Capital for regulatory purposes.

\end{itemize}
There have been surprisingly few papers dealing with cost of capital issues in a regulatory context in the last fifty years. Following a key article by \cite{myers1972application}, the use of the Capital Asset Pricing Model (CAPM) became widespread in public utility regulatory practice.\footnote{This adoption of the CAPM was not without debate. \cite{breen1972use} question the utility of the CAPM model in regulatory processes, prompting the response \cite{myers1972use}. See also \cite{brigham1977use} and \cite{peseau1978use}. \cite{litzenberger1980capm} provides a useful summary of the different forms of CAPM and some early issues that arose in its application to public utilities. \cite{chartoff1982case} provide a critique of CAPM from a legal perspective.} Over the next few decades, as advances in corporate finance theory developed new tools for estimation of the cost of capital (including arbitrage pricing theory, discounted cash flow, multi-factor models, or dividend growth models), these were proposed and considered in the context of public utility regulation (see, e.g., the surveys in \cite{kolbe1984cost, alexander2000few,villadsen2017risk}).\footnote{\cite{roll1983regulation} advocate for the use of the arbitrage pricing theory in the context of public utility regulation.} On the whole, however, public utility regulators have tended to adopt developments in corporate finance theory.\footnote{One exception is \cite{michelfelder2013public}.}

This article has four main sections. Section \ref{sec:prelim} introduce some fundamental principles of cost of capital, for both individual cash-flows and streams of cash-flows. Section \ref{sec:rp1y} starts with the simplest case of a regulatory period that lasts one year. This section introduces the circularity problem and shows how it can be resolved by using a different cost of capital for the components of the cash-flow of the firm. Section \ref{sec:rpmy} extends that analysis to a regulatory period of $T$ years. Section \ref{sec:cod} explores the issues associated with estimating the cost of debt and argues that there is no value (and it may not be possible) to estimate a separate cost of equity and cost of debt. Section \ref{sec:con} concludes.



\section{Preliminaries}\label{sec:prelim}

As set out below, we will define the concept of the cost of capital with reference to the properties of the `valuation functional' or `valuation operator'. We will begin, therefore, by setting out the properties of this functional.

\subsection*{Definition of the present-value functional}

Let's suppose we are currently at time 0 and there is an uncertain future cash-flow arriving at time $t$. The uncertain payoff of this cash-flow at time $t$ is assumed to be reflected in the value of the random variable $X_t$. The \textbf{expected value} of $X_t$, in the light of the information available at time zero, denoted $\E_0(X_t)$, is a function (technically, a functional) from a set of random variables representing payoffs at time $t$ to the expected value (or mean) of those payoffs, when viewed from time zero.

Similarly, the \textbf{present value} of the future cash-flow at time zero, denoted $\V_{0 \to t}(X_t)$ is a function from the set of random variables representing payoffs at time t, to the present value of that random variable at time zero \citep{skiadis2022, hansen2009long, anderson2012, ross1978simple}. The present value of a future cash-flow at time 0 is also the price at which the right to receive that future cash-flow would trade at time 0 in an efficient market.

Both the expected value $\E_0(\cdot)$ and the present value $\V_{0 \to t}(\cdot)$ of a cash-flow are linear functions. If $X_t$ and $Y_t$ are both cash-flows arriving in at time $t$, and $a$ and $b$ are real numbers, the present value function satisfies:
\begin{equation}
    \V_{0 \to t}(a X_t + b Y_t) = a \V_{0 \to t}(X_t)+ b \V_{0 \to t}(Y_t)
\end{equation}


When a cash-flow arrives further in the future (say, in period 2), with uncertain payoff $X_2$, the present value of the cash-flow at time 1 is denoted $\V_{1 \to 2}(X_2)$. Viewed from time zero, this value is itself a random variable. The present value functional satisfies a recursion property:
\begin{equation}
    \forall 1\le t \le T,\;\; \V_{0 \to t}(\V_{t \to T}(X_T))=\V_{0 \to T}(X_T)
\end{equation} 

For regulatory purposes we often need to determine the present value of an indefinite \textit{stream of cash-flows} $X_1, X_2, X_3 , \ldots$, arriving at time 1, 2, 3 and so on. The present value of a stream of cash-flows is just the sum of the present value of the individual cash-flows.
\begin{equation}
    \V_0(X_1, X_2, X_3, \ldots) = \sum_{t=1} \V_{0 \to t}(X_t)
\end{equation}

Using the recursion property of the present value function, the present value of an indefinite stream of cash-flows is equal to the present value of a finite stream of cash-flows, truncated in year $T$, say, where we simply add the `terminal value' of the cash-flow stream to the final individual cash-flow:
\begin{align}
    \V_0(X_1, X_2, X_3, \ldots) &= \sum_{t=1} \V_{0 \to t}(X_t) \\
    &=\sum_{t=1}^{T} \V_{0 \to t}(X_t) + \sum_{t=T} \V_{0 \to t}(X_t)\\
    &= \sum_{t=1}^{T} \V_{0 \to t}(X_t)+ \sum_{t=T+1} \V_{0 \to T}(\V_{T \to t}(X_t))\\
    &=\sum_{t=1}^{T} \V_{0 \to t}(X_t) + \V_{0 \to T}( \V_T(X_{T+1}, X_{T+2}, \ldots)\\
    &=\V_0 (X_1, X_2, \ldots, X_{T-1}, X_T+\V_T(X_{T+1}, X_{T+2}, \ldots))
\end{align}
In the case where $T=1$, the present value of a stream of cash-flows is the sum of two terms: The one-period cash-flow $X_1$ and the terminal value at time 1 $\V_1=\V_1(X_2, X_3, X_4, \ldots)$:
\begin{equation}
    \V_0=\V_0(X_1+\V_1)
\end{equation}

\subsection*{Definition of the cost of capital}

The \textbf{cost of capital} for a cash-flow represented by the random variable $X_t$, denoted $\R_{0 \to t}(X_t)$ is the ratio of the expected value of the cash-flow to the present value (provided the present value of the cash-flow is not zero, in which case the cost of capital is undefined).\footnote{It is also common to define the cost of capital as equation \ref{eqn:cocdef} minus one. We have a slight preference for the approach set out here as it saves having to add and/or subtract one from different equations.}
\begin{equation}
    \R_{0 \to t}(X_t) \equiv \frac{\E_0(X_t)}{\V_{0 \to t}(X_t)}\label{eqn:cocdef}
\end{equation}

The present value of a cash-flow arriving at time zero, viewed from time zero, is just the expected value:
\begin{equation}
    \V_{0 \to 0}(X_0) = \E_0(X_0)
\end{equation}
It follows that the cost of capital at time zero for any cash-flow arriving at time zero is equal to one.\footnote{More generally, of course, $\R_{t \to t}(X_t)=1$.}

The present value of a certain cash-flow arriving at time $t$ in the future is the ratio of the certain cash-flow to a value known as the `risk-free rate'. Let's suppose that the cash-flow $X_t$ yields the certain value $a$ at time $t$, then:
\begin{equation}
    \V_{0 \to t}(X_t) = \frac{a}{RF_{0 \to t}}
\end{equation}
Here $RF_{0 \to t}$ is the discount rate or interest rate for a certain cash-flow at time $t$, viewed from time zero. It follows that the cost of capital for a cash-flow with a fixed payoff $a$ arriving at time $t$ is just the risk-free rate between time zero and time $t$:
\begin{equation}
    \R_{0\to t}(a)=\frac{\E_0(a)}{\V_{0 \to t}(a)}=RF_{0 \to 1t}
\end{equation}.

Due to the linearity properties of the expectation and value operators, the cost of capital of any cash-flow is \textit{independent of the scale} of the cash-flow:\footnote{It follows that we need only define the cost of capital for uncertain cash-flows with a mean of one.}
\begin{equation}
    \R_{0 \to t}(a X_t)=\frac{\E_0(a X_t)}{\V_{0 \to t}(a X_t)}=\R_{0 \to t}(X_t)\; \text{for any} \; a \ne 0\label{eqn:scale}
\end{equation}

Now consider the cost of capital for the sum of two cash-flows $X_t$ and $Y_t$. It follows from the linearity properties of the present value function that the cost of capital for the sum of two cash-flows can be written as the \textit{weighted average} of the cost of capital of the individual cash-flows. If $X_t$ and $Y_t$ are both cash-flows arriving at time $t$, the cost of capital for the sum can be expressed as a function of the cost of capital of the components in two ways. In the first formulation, the weighting depends on the share of the present value of each component in the sum:
\begin{equation}
    \R_{0 \to t}(X_t + Y_t) = \alpha \R_{0 \to t}(X_t)+ (1-\alpha) \R_{0 \to t}(Y_t)\label{eqn:wa1}
\end{equation}
Where:
\begin{equation}
    \alpha=\frac{\V_{0 \to t}(X_t)}{\V_{0 \to t}(X_t+Y_t)} \label{eqn:wb1}
\end{equation}
In the second formulation, the weighting depends on the share of the expected value of each component in the sum:
\begin{equation}
    \R_{0 \to t}(X_t + Y_t)^{-1} = \alpha \R_{0 \to t}(X_t)^{-1}+ (1-\alpha) \R_{0 \to t}(Y_t)^{-1}\label{eqn:wa2}
\end{equation}
Where:
\begin{equation}
    \alpha=\frac{\E_{0 \to t}(X_t)}{\E_{0 \to t}(X_t+Y_t)} \label{eqn:wb2}
\end{equation}

Finally, let's suppose we have a cash-flow $X_T$ arriving at time $T$. As we have seen, the present value of this cash-flow at an earlier time $t$ is denoted $V_{t \to T}(X_T)$. This is also a random variable. The cost of capital for this random variable at time zero is the ratio of the expected future price $\E_0(\V_{t \to T}(X_T))$ to the current price $\V_{0 \to T}(X_T)$:\footnote{There is also a formula for the cost of capital for a cash-flow arriving in two periods, as a function of the one-period cash-flows from time 0 to time 1 and from time 1 to time 2, as follows: $\R_{0\to 2}(X_2) = \R_{0\to 1}(\V_1(X_2)) \E_0(X_2) \E_0( \frac{\E_1(X_2)}{\R_{1\to 2}(X_2)})^{-1}$.}
\begin{equation}
    \R_{0 \to t}(\V_{t \to T}(X_T)) = \frac{\E_0(\V_{t \to T}(X_T))}{\V_{0 \to t}(\V_{t \to T}(X_T))}= \frac{\E_0(\V_{t \to T}(X_T))}{\V_{0 \to T}(X_T)}
\end{equation}

This approach is independent of the methodology used to value cash-flows. But it may be worth pointing out that, with some additional assumptions, the Capital Asset Pricing Model emerges from this approach. In the special case where the preferences of the individual investors take the mean-variance form, the cost of capital for any cash-flow $X_1$ arriving at time 1 is given by an equation which is analogous to the standard Capital Asset Pricing Model. This result is proved in appendix \ref{app:a}.
\begin{equation}
    \R_{0 \to 1}(X_1)^{-1} = RF_{0 \to 1}^{-1} - (RF_{0 \to 1}^{-1} - \R_{0 \to 1}(M_1)^{-1}) \beta_{0 \to 1}(X_1)\label{eqn:capm}
\end{equation}
Where $M_1$ is the cash-flow from holding the `market portfolio' at time 1, and $\beta_{0 \to 1}(X_1)$ is the `beta' of the cash-flow $X_1$ defined as:
\begin{equation}
    \beta_{0 \to 1}(X_1)=\frac{\E_0(M_1)}{\E_0(X_1)}\frac{Cov(X_1,M_1)}{Var(M_1)}
\end{equation}
Equation \ref{eqn:capm} is the correct formulation of the CAPM in this context. In this article, however, we will present more general results derived from the properties of the value functional, rather than a special case (such as the CAPM).

\section{One-year regulatory period}\label{sec:rp1y}

Let's now turn to examine questions regarding the estimation of cost of capital in the context of regulatory proceedings. At the outset it is important to be clear that there is no single correct definition of the cost of capital in all regulatory proceedings. There are many ways to regulate that achieve the overall objective of NPV=0. The correct cost of capital for any given regulatory process depends on precisely how the regulatory proceeding is formulated. Therefore we must be clear exactly how the regulatory process operates.

In order to keep things simple, let's start by considering the issues that arise when the regulatory period (that is, the period between price reviews or `regulatory resets') lasts exactly one year. This yields the simplest and most familiar form of the regulatory process.

\subsection*{The standard regulatory process when the regulatory period lasts one year}

We will assume that the regulator follows a standard regulatory process, known in Australia as the `Building Block Model'. In its most general form this is described as follows (and illustrated in figure \ref{fig:1}):
\begin{enumerate}
    \item At the start of each regulatory period (which we will label time zero) the regulator observes the value of the opening asset base $RAB_0$, forecasts opex, capex and sales during the regulatory period, and chooses a set of prices to apply during the regulatory period and the closing asset base $RAB_1$.
    \item At the end of the regulatory period (time one), the out-turn values of opex, capex and revenue (and therefore the cash-flow of the firm $X_1$) are realised, together with the closing asset base. This becomes the opening asset base for the subsequent regulatory period.
\end{enumerate}
This process continues over the life of the firm. At the end of the life of the firm the regulator ensures that the closing regulatory asset base is equal to zero.

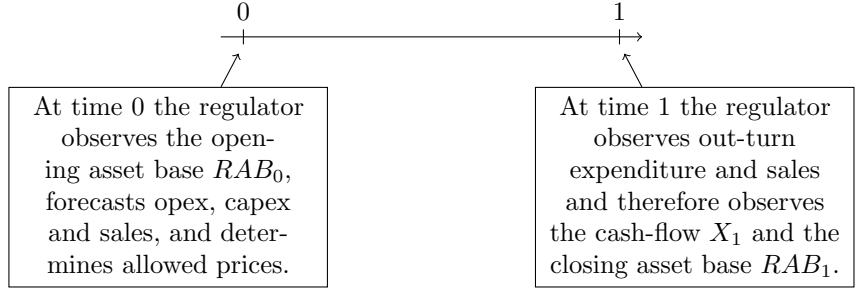
\begin{figure}[ht]
\singlespacing
    \centering
    \caption{Time evolution of a standard one-year regulatory process}
    \begin{tikzpicture}
    \draw [->] (-0.3,0) to (5.3,0);
    \draw (0,-0.1) to (0,0.1);
    \draw (5,-0.1) to (5,0.1);
    \node [above] (A) at (0,0.1) {0};
    \node [above] (B) at (5,0.1) {1};
    \node (A1) at (0,-0.1) {};
    \node (B1) at (5,-0.1) {};
    \node [draw, text width=4 cm, align = center] (C) at (-1,-2) {At time 0 the regulator observes the opening asset base $RAB_0$, forecasts opex, capex and sales, and determines allowed prices.};
    \node [draw, text width=4 cm, align = center] (D) at (6,-2) {At time 1 the regulator observes out-turn expenditure and sales and therefore observes the cash-flow $X_1$ and the closing asset base $RAB_1$.};
    \draw [->] (C) to (A1);
    \draw [->] (D) to (B1);
    \end{tikzpicture}
    \label{fig:1}
\end{figure}

Formally the regulatory process operates as follows: At the start of each regulatory period the regulator chooses the cash-flow of the firm $X_1$ and the closing asset base $RAB_1$ in such a way that the present value of the sum of these two is equal to the opening asset base $RAB_0$:
\begin{equation}
    RAB_0=\V_0(X_1+RAB_1)\label{eqn:bbm}
\end{equation}

 Under these assumptions the Fundamental Theorem of Regulation (see Appendix \ref{app:b}) shows that, at each point in time the asset base is equal to the future stream of cash-flows $RAB_t=\V_t(X_{t+1}, X_{t+2}, \ldots)$ and the regulated firm achieves NPV=0.

We will refer to $X_1$ as the \textbf{one-period cash-flow} of the firm, and the sum $X_1+RAB_1$ as the \textbf{combined cash-flow}.

Equation \ref{eqn:bbm} may look at first abstract, but we can move a step closer to regulatory practice by expanding the present value using the definition of the cost of capital and re-writing the equation as an expression for the allowed level of the cash-flow:
\begin{align}
    RAB_0 &= \V_0(X_1+RAB_1) = \frac{\E_0(X_1+RAB_1)}{\R_{0 \to 1}(X_1+RAB_1)}\nonumber\\
    \implies \; \E_0(X_1) &= \R_{0 \to 1}(X_1+RAB_1) \times RAB_0 -\E_0(RAB_1)\label{eqn:bbm9}
\end{align}

This equation can be made to look more familiar by adopting the conventional definition of the cost of capital $1+r_{0\to t}(X_t)=\R_{0 \to t}(X_t)$ and expanding the cash-flow allowance into its components:
\begin{equation}
X_1 = \underbrace{R_1}_{\text{Revenue}} - \underbrace{O_1}_{\text{Opex}} - \underbrace{K_1}_{\text{Capex}}
\end{equation}
With these definitions, equation \ref{eqn:bbm9} can be written in the following way: The expected revenue allowance\footnote{The revenue $R_1$ here should not be confused with the cost of capital $\R_{0 \to 1}$, for which we will use the blackboard bold font.} can be written as the normal sum of opex, `return on' and `return of' capital:
\begin{equation}
\E_0(R_1) = \underbrace{r_{0 \to 1}(X_1+RAB_1) \times RAB_0}_{\text{`Return on capital'}} + \underbrace{\E_0(O_1)}_{\text{Opex}} + \underbrace{\E_0(Dep_{0 \to 1})}_{\text{`Return of capital'}}\label{eqn:bbm10}
\end{equation}
Here:
\begin{equation}
Dep_{0 \to 1} = RAB_0 + K_1 - RAB_1\label{eqn:bbm11}
\end{equation}
Equations \ref{eqn:bbm10} and \ref{eqn:bbm11} are the familiar two equations which define the Building Block Model (the standard regulatory process used in Australia and around the world).

\subsection*{Problems in the estimation of the cost of capital}

The previous section noted that equation \ref{eqn:bbm9} (or its more familiar variants, equations \ref{eqn:bbm10} and \ref{eqn:bbm11}) captures the standard formulation of the regulatory process. But it is already apparent that there is a problem. As equation \ref{eqn:bbm9} shows, the fundamental task of the regulator is to choose a set of regulated prices so that the value for the expected cash-flow $X_1$ satisfies equation \ref{eqn:bbm9}. But $X_1$ appears on both sides of equation \ref{eqn:bbm9}. The regulator cannot obtain a formally correct cash-flow allowance $\E_0(X_1)$ (or cost of capital $\R_{0 \to 1}(X_1+RAB_1)$) without knowing how the cost of capital depends on the cash-flow allowance.\footnote{A similar circularity problem can arise when estimating the cost of capital using a weighted average formula: the weightings on the cost of equity and the cost of debt depend, in part on the value of the firm, which depends on the cost of capital. See \cite{mohanty2003practical}.}

In addition, the cost of capital $\R_{0 \to 1}(X_1+RAB_1)$ will, in general, depend on the choice of the closing asset base $RAB_1$.\footnote{In cases where the closing asset base is stochastic the cost of capital may also depend on the variation in the closing asset base.} Because the regulatory asset base varies over the life of the firm, we can expect that the cost of capital $\R_{0 \to 1}(X_1+RAB_1)$ also varies over the life of the firm, even if all other factors in the environment remain constant.

The standard resolution of this problem in regulatory practice is just to assume that the cost of capital $\R_{0 \to 1}(X_1+RAB_1)$ is independent of the cash-flow allowance $\E_0(X_1)$ and independent of the choice of the closing asset base $RAB_1$. This is a heuristic which is virtually universal in regulatory practice.

Is there a formulation of the regulatory process which does not suffer from this circularity problem? It turns out that there is. We can see this by returning to equation \ref{eqn:bbm}, and re-writing it as follows:
\begin{align}
    RAB_0 &=\V_0(X_1)+\V_0(RAB_1)\nonumber\\
    &=\frac{\E_0(X_1)}{\R_{0 \to 1}(X_1)}+\frac{\E_0(RAB_1)}{\R_{0 \to 1}(RAB_1)}
\end{align}
Therefore, we can re-formulate the Building Block Model (equation \ref{eqn:bbm9}) in a way which resolves the circularity problem:\footnote{As noted in equation \ref{eqn:scale} the cost-of-capital for the one-period cash-flow $\R_{0 \to 1}(X_1)$ is independent of the level of the expected cash-flow $\E_0(X_1)$.}
\begin{equation}
    \E_0(X_1)= \R_{0 \to 1}(X_1) \times RAB_0 -\E_0(RAB_1) \frac{\R_{0 \to 1}(X_1)}{\R_{0 \to 1}(RAB_1)}\label{eqn:b3}
\end{equation}

Comparing equations \ref{eqn:b3} and \ref{eqn:bbm9} we can see that this version of the Building Block Model is similar to the standard approach  except for the following:
\begin{itemize}
    \item The `return on capital' is defined by multiplying the cost of capital for the one-period cash-flow $\R_{0 \to 1}(X_1)$ (rather than the cost of capital for the combined cash-flow $\R_{0 \to 1}(X_1+RAB_1)$) by the opening asset base; and
    \item In the `return of capital', the closing asset base is scaled by the ratio of the costs-of-capital of the components (i.e., the ratio $\R_{0 \to 1}(X_1)/\R_{0 \to 1}(RAB_1)$.
\end{itemize}

We are now in a position to answer the first two of the questions set out above. Does the cost of capital for a regulated firm vary with the level of the regulatory asset base or the level of the allowed cash-flow?

The answer is as follows: In the conventional historic regulatory practice in Australia, the relevant cost of capital (that is, the multiplier of the asset base in the revenue allowance equation) is the cost of capital for the combined cash-flow $\R_{0 \to 1}(X_1+RAB_1)$. This will, in general vary with both the expected level of the one-period cash-flow $\E_0(X_1)$ and the expected level of the closing asset base $\E_0(RAB_1)$ even if nothing else in the environment changes.

In particular, the cost of capital will depend on the ratio of the level the one-period cash-flow relative to the closing asset base. If this ratio is small, the correct cost of capital is closer to the cost of capital for the closing asset base alone $\R_{0 \to 1}(RAB_1)$ which would normally be expected to be close to the risk-free rate. If this ratio is large, the correct cost of capital is closer to the cost of capital for the one-period cash-flow $\R_{0 \to 1}(X_1)$ which could be very large. As the ratio changes the relevant cost of capital will change, even if there are no other changes in the environment. There are examples of this effect in table \ref{table:1} and figure \ref{fig:2} below.

The dependence of the relevant cost of capital on the level of the cash-flow gives rise to a problem of circularity in the regulatory process. This problem of circularity can be resolved by changing the formulation of the Building Block Model to the formulation set out in equation \ref{eqn:b3} in which the relevant cost of capital is the cost of capital for the one-period cash-flow alone. This cost of capital does not depend on the level of the cash-flow.

\subsection*{Do these concerns make a difference in practice?}

The sections above have identified possible concerns with the conventional approach to setting the cost of capital. But do these issues make a material difference in practice?

To make this assessment let's explore how much the cost of capital of the combined cash-flow $\R_{0 \to 1}(X_1+RAB_1)$ might vary even if the costs of capital on the component cash-flows $\R_{0 \to 1}(X_1)$ and $\R_{0 \to 1}(RAB_1)$ remain unchanged.

Let's suppose that the cost of capital for the cash-flow $X_1$, is say, 20\%, and the cost of capital for the closing asset base $RAB_1$ is, say, 5\%, so $\R_{0 \to 1}(X_1)=1.2$ and $\R_{0 \to 1}(RAB_1)=1.05$. Table \ref{table:1} sets out the cash-flow allowance and the relevant cost of capital under different assumptions about the level of the opening and closing asset base.\footnote{The fourth column in table \ref{table:1} is given by equation \ref{eqn:b3}; the fifth column is given by equation \ref{eqn:wa2}.} As can be seen, the cost of capital for the combined cash-flow $X_1+RAB_1$ varies widely with the level of the asset base, even though there is no change in the underlying systematic risk faced by the firm (that is, no change in $\R_{0 \to 1}(X_1)$ and $\R_{0 \to 1}(RAB_1)$).

\newcommand\Tstrut{\rule{0pt}{2.6ex}}         
\newcommand\Bstrut{\rule[-0.9ex]{0pt}{0pt}}   

\begin{table}[h!]
\singlespacing
\centering
\caption{The combined cost of capital for a regulated firm $\R_{0 \to 1}(X_1+RAB_1)$ varies with the size of the components even if there is no change in the cost of capital for the components}
\begin{tabular}{|c | c | c | c | c | c |} 
 \hline 
 $RAB_0$ & $\E_0(RAB_1)$ & $\R_{0 \to 1}(X_1)$ & $\R_{0 \to 1}(RAB_1)$ & $\E_0(X_1)$ & $\R_{0 \to 1}(X_1+RAB_1)$ \Tstrut\Bstrut\\
 \hline
 \$1,000 & \$900 & 20\% & 5\% & \$171.83 & 7.14\%\Tstrut\\ 
 \$500 & \$400 & 20\% & 5\% & \$142.86 & 8.57\% \\
 \$1,000 & \$800 & 20\% & 5\% & \$285.71 & 8.57\%\\
 \$400 & \$200 & 20\% & 5\% & \$251.73 & 12.86\% \Bstrut\\ 
 \hline
\end{tabular}
\label{table:1}
\end{table}

Figure \ref{fig:2} presents another way of looking at this question. The left hand graph in figure \ref{fig:2} illustrates the impact on the combined cost of capital $\R_{0 \to 1}(X_1+RAB_1)$ of changes in the allowed regulatory cash-flow $\E_0(X_1)$ holding all other factors constant (here we assume $\E_0(RAB_1)=\$1000$ and  $\R_{0 \to 1}(X_1)=1.2$ and $\R_{0 \to 1}(RAB_1)=1.05$ as before). As can be seen, an increase in the allowed regulatory cash-flow (perhaps due to, say, a reduction in forecast expenditure) results in an increase in the combined cost of capital, even if nothing else changes in the environment.

The right hand graph in figure \ref{fig:2} illustrates how the combined cost of capital varies over the life of a firm with changes in the asset base. This graph illustrates the case of a firm which starts with an opening asset base of \$1000, and lasts five years. The regulated cash-flow allowance is chosen to be constant at $\E_t(X_{t+1})=\$263.97$ each year, ensuring that the closing asset base at the end of the life of the firm is zero ($RAB_5=0$). The regulatory asset base starts at \$1000 and declines to zero over the five years. As can be seen, the combined cost of capital $\R_{t \to t+1}(X_{t+1}+RAB_{t+1})$ increases as the regulatory asset base declines.

In both of these graphs the cost of capital for the one-period cash-flow is fixed at 20\% ($\R_{t \to t+1}(X_{t+1})=1.20$) and the cost of capital for the closing asset base is fixed at 5\% ($\R_{t \to {t+1}}(RAB_{t+1})=1.05$). 

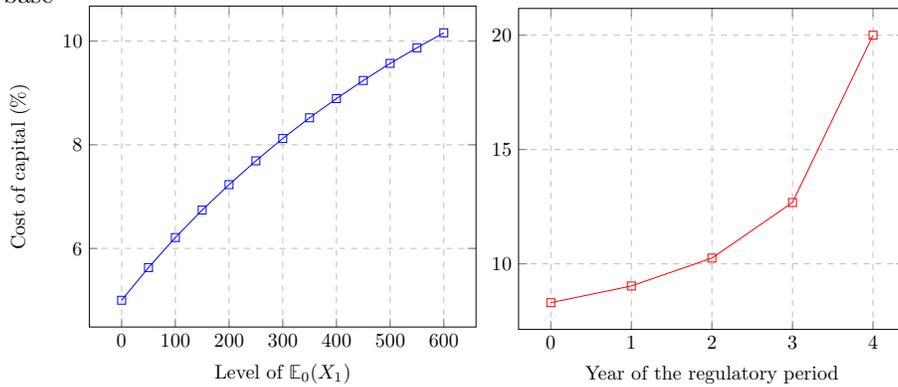
\begin{figure}[ht]
\singlespacing
    \centering
    \caption{Illustration of the dependence of the combined cost of capital $\R_{0 \to 1}(X_1+RAB_1)$ on the level of the allowed cash-flow or the regulatory asset base}
    \begin{tikzpicture}[scale=0.75]
        \begin{axis} [
            xlabel={Level of $\E_0(X_1)$},
            ylabel={Cost of capital (\%)},
            ymajorgrids=true,xmajorgrids=true,
            grid style=dashed,
        ]
        \addplot[
            color=blue,
            mark=square,
            ]
            coordinates {
            (0,5.00)(50,5.63)(100,6.21)(150,6.74)(200,7.23)(250,7.69)(300,8.12)(350,8.52)(400,8.89)(450,9.24)(500,9.57)(550,9.87)(600,10.16)
            };
        \end{axis}
        \end{tikzpicture}
    \begin{tikzpicture}[scale=0.75]
        \begin{axis} [
            xlabel={Year of the regulatory period},
            ymajorgrids=true,xmajorgrids=true,
            grid style=dashed,
        ]
        \addplot[
            color=red,
            mark=square,
            ]
            coordinates {
            (0,8.3)(1,9.03)(2,10.25)(3,12.68)(4,20)
            };
        \end{axis}
        \end{tikzpicture}
   \label{fig:2}
\end{figure}

\section{Multi-year regulatory period}\label{sec:rpmy}

For many years, the standard regulatory practice in Australia and New Zealand has involved the use of a fixed-length (five-year) regulatory period, with a single cost of capital for that period. Perhaps unsurprisingly, it turns out that this practice substantially complicates the question of setting the appropriate cost of capital.


As before, the relevant cost of capital depends very heavily on the precise formulation of the regulatory process. Therefore, as before, we must be precise as to the regulatory process we are using.

Let's suppose that the regulatory period lasts $T$ years (where $T>1$). In the context of a multi-year regulatory period, equation \ref{eqn:bbm} generalises as follows. Given the opening asset base $RAB_0$, the task of the regulator is to choose the cash-flow in each year of the regulatory period $X_1$, $X_2$, \ldots, $X_T$, and the closing asset base $RAB_T$, to satisfy the following expression:
\begin{equation}
    RAB_0=\V_0(X_1, X_2, \ldots, X_T+RAB_T)=\sum_{t=1}^T \V_0(X_t)+\V_0(RAB_T)\label{eqn:c1}
\end{equation}

In common regulatory practice, this is usually implemented as follows. A single parameter, which we will label $R$ is chosen. In addition, a sequence of values of the asset base $RAB_t$ is chosen. The regulated cash-flow allowance in each year of the regulatory period is then determined using a simple analogy to equation \ref{eqn:bbm9}:
\begin{align}
    \forall t=1, \ldots, T,\;\; \E_0(X_t) &=  R \times RAB_{t-1} -\E_0(RAB_t)\label{eqn:c2}
\end{align}

But, how should we choose the parameter $R$? Expanding out equations \ref{eqn:c2} and applying equation \ref{eqn:c1} we find that, given the expected cash-flows $\E_0(X_t)$ and the opening asset base $RAB_0$ and the expected closing asset base $\E_0(RAB_T)$, the parameter $R$ must be chosen to satisfy the following:
\begin{align}
    &\frac{\E_0(X_1)}{R}+\frac{\E_0(X_2)}{R^2}+\ldots+\frac{\E_0(X_T+RAB_T)}{R^T}=RAB_0\label{eqn:c4}
\end{align}
In other words, in this approach to regulation, the correct value for the cost of capital parameter $R$ is the \textbf{internal rate of return} of the cash-flow stream consisting of an outlay of $RAB_0$ at time zero, and cash-flow of $\E_0(X_1), \E_0(X_2), \ldots, \E_0(X_T+RAB_T)$ in the subsequent years of the regulatory period.

But again, we can see that there is a problem. As equation \ref{eqn:c4} makes clear, the parameter $R$ depends on the levels of the individual cash-flows $\E_0(X_1)$, $\E_0(X_2)$ and so on, in a complicated manner. But the parameter $R$ is also a key input into the determination of the $\E_0(X_1)$ and so on through equation \ref{eqn:c1}. Once again we have a problem of circularity.


As before, this problem can be resolved by changing the regulatory process. Rather than using a single cost of capital parameter $R$ for the entire regulatory period we should use different costs of capital for the individual components of the cash-flow of the firm.

There are two ways that this might be carried out. In the first approach, the allowed cash-flows are all set in advance, based on the costs of capital prevailing at the start of the regulatory period. In the second approach, the allowed cash-flows are set each year of the regulatory period, on the basis of the costs of capital prevailing at that time.

Let's consider the first approach in which the allowed cash-flows are all set in advance. Specifically, let's suppose that the regulator chooses parameters $R_t$ and $S_t$ for $t=1, \ldots, T$. The regulator then sets the cash-flow allowance as follows (this equation is the analogy to equation \ref{eqn:b3}):
\begin{align}
    \forall t=1, \ldots, T,\;\; \E_0(X_t) &=  R_t \times \E_0(RAB_{t-1}) -\E_0(RAB_t)\frac{R_t}{S_t}\label{eqn:d3}
\end{align}
The cash-flows set in this way satisfy the fundamental requirement of equation \ref{eqn:c1} provided we choose $R_1=\R_{0 \to 1}(X_1)$, $S_1 R_2=\R_{0 \to 2}(X_2)$, \ldots , $S_1 S_2 \ldots S_T = \R_{0 \to T}(RAB_T)$. These equations can be satisfied in different ways. However, a straightforward approach is to choose:\footnote{$S_t$ here is the `forward rate' for the cost of capital for the asset base -- specifically, it is the rate at time zero that is forecast to apply between time $t-1$ and time $t$.}
\begin{align}
    R_t &= \frac{R_{0 \to t}(X_t)}{\R_{0 \to {t-1}}(RAB_{t-1})}\label{eqn:rt}\\
    \text{and} \; S_t&=\frac{\R_{0 \to t}(RAB_t)}{\R_{0 \to {t-1}}(RAB_{t-1})}\label{eqn:st}
\end{align}
It is straightforward to check that, with this choice of the parameters, equation \ref{eqn:c1} is satisfied.

Under the second approach, the allowed cash-flows are not fixed in advance, but are set at the start of each year (within the regulatory period) on the basis of information that is available at the time. In this case, the relevant equation for establishing the cash-flow allowance can be written as follows:
\begin{align}
    \forall t=1, \ldots, T,\;\; \E_{t-1}(X_t) &=  R_t \times RAB_{t-1} -\E_{t-1}(RAB_t)\frac{R_t}{S_t}\label{eqn:d4}
\end{align}
This equation is the generalisation of equation \ref{eqn:b3}. The formulae to calculate the relevant values of the parameters $R_t$ and $S_t$ in this case are derived in appendix \ref{app:c}.

To see how the first approach might work in a simple example, let's suppose that we have a regulatory period that lasts five years. At time zero the opening asset base is $RAB_0=\$1,000$. The regulator chooses a path for the asset base $RAB_1=\$900$, $RAB_2=\$800$, \ldots, $RAB_5=\$500$. The cost of capital for the cash-flow and for the asset base (which is just equal to the risk-free rate of 5\% per annum) are set out in table \ref{table:3}. The cost of capital for the asset base (assuming a flat term structure and a risk-free rate of 5\% per annum) is:
\begin{equation}
    \R_{0 \to t}(RAB_t) = (1.05)^t
\end{equation}
The cost of capital for the cash-flow is chosen to satisfy the condition:\footnote{This can be justified using the assumptions that the term structure is flat and no new information about the future cash-flow arrives over time.}
\begin{equation}
    \R_{0 \to t}(X_t) = (1.05)^{t-1}(1.20)
\end{equation}
The remainder of the table shows the implied values of $R_t$ and $S_t$ and the resulting cash-flow allowance $\E_0(X_t)$.\footnote{$R_t$ and $S_t$ are given by equations \ref{eqn:rt} and \ref{eqn:st} respectively. $\E_0(X_t)$ is given by equation \ref{eqn:d3}.}

Using these values we can calculate that the value of the parameter $R$ (the single cost of capital for entire regulatory period) is 7.16\%. As before, this value is a complicated mix of the different costs of capital for the different cash-flows and timings during the regulatory period. It is not possible to calculate this value until after the regulatory cash-flow allowances have been determined.

The last row of table \ref{table:3} shows the implied value for the cost of capital that would be required if the regulator followed the naive approach of equation \ref{eqn:bbm9}. As can be seen, in this case the parameter chosen is a `mixture' of cost of capital for the cash-flow $X_t$ and the asset base $RAB_t$ and varies across the regulatory period.

\begin{table}[h!]
\singlespacing
\centering
\caption{Illustration of calculating the cash-flow allowance over a five-year regulatory period}
\begin{tabular}{|c | c | c | c | c | c | c |} 
 \hline
 Time & $t=0$ & $t=1$ & $t=2$ & $t=3$ & $t=4$ & $t=5$\Tstrut\Bstrut\\ 
 \hline
 $RAB_t$ & \$1,000 & \$900 & \$800 & \$700 & \$600 & \$500\Tstrut\\ 
 $\R_{0 \to t}(X_t)$ & 1.000 & 1.200 & 1.260 & 1.323 & 1.389 & 1.459 \\
 $\R_{0 \to t}(RAB_t)$ & 1.000 & 1.05 & 1.103 & 1.158 & 1.216 & 1.276 \\
 $R_t$ &  & 1.20 & 1.20 & 1.20 & 1.20 & 1.20\\
 $S_t$ &  & 1.05 & 1.05 & 1.05 & 1.05 & 1.05\\
 $\E_0(X_t)$ & & \$171.43 & \$165.71 & \$160.00 & \$154.29 & \$148.57 \\
 Implied CoC & & 7.14\% & 7.30\% & 7.50\% & 7.76\% & 8.10\% \Bstrut\\ 
 \hline
\end{tabular}
\label{table:3}
\end{table}

At this point we can answer one of the questions posed at the outset: What is the correct term for the cost of capital used in regulatory proceedings? From the analysis above we can provide the following answer:
\begin{itemize}
    \item Where a single cost-of-capital parameter is used to determine all of the cash-flow allowances throughout a regulatory period (as in equation \ref{eqn:c2}), there is no single correct `term' for this cost-of-capital. Rather, this cost of capital is a mix of a number of different terms (corresponding to the cash-flows within the regulatory period, and the asset base at the end of the regulatory period). The term of each of these individual cash-flows is shorter than or equal to the length of the regulatory period. Formally it is the internal rate of return associated with a cash-flow stream.
    \item As before there arises a circularity problem. The single cost-of-capital parameter depends on the cash-flow allowances, but the cash-flow allowances depend on this cost-of-capital parameter. In order to resolve the circularity problem we must use a separate cost of capital for each of the individual components of the cash-flow of the firm -- that is, a separate cost of capital for $X_t$, $t=1, \ldots, T$ and $RAB_T$.
    
    In the case where all of the cash-flow allowances are set at the beginning of the period (the first approach above), the regulatory process must follow equation \ref{eqn:d3}. Specifically:
    \begin{itemize}
        \item The relevant cost of capital for the purposes of determining the allowed `return on capital' (that is, the coefficient on the regulatory asset base) must be given by:
        \begin{equation}
            R_t=\frac{R_{0 \to t}(X_t)}{\R_{0 \to {t-1}}(RAB_{t-1})}
        \end{equation}
        \item In calculating the `return of capital' the closing regulatory asset base must be scaled by
        \begin{equation}
            \frac{R_t}{S_t}=\frac{R_{0 \to t}(X_t)}{\R_{0 \to {t}}(RAB_t)}
        \end{equation}
        This reduces to equation \ref{eqn:b3} in the case of a regulatory period of one year.
    \end{itemize}
    
\end{itemize}

As an illustration of the effect of these observations, let's use some typical values for the cash-flow and asset base drawn from an actual regulatory proceeding. In this case we will use distribution businesses in Australia. These are regulated using a five-year regulatory cycle. From the discussion above, this means that there are 6 costs-of-capital we need to estimate: $\mathbb{R}_{0 \to t}(X_t)$, $t=1,\ldots,5$ and $\mathbb{R}_{0 \to 5}(RAB_5)$. We will assume values for these costs of capital and then determine the implications for the value of the parameter $R$.

As before we will assume that the (annualised) risk-free interest rate for all terms is, say, 1.05 (5\%). In other words, the cost of capital for a fixed value in one year is $1.05$, in two years is $1.05^2$, and so on. Following standard regulatory practice, we will make the assumption that $RAB_5$ is a fixed number, so it should receive a cost of capital equal to the risk-free rate which as just noted, is $1.05^5$. We will make the assumption that, at time 0, $\mathbb{V}_{s \to t}(X_t)$ is a fixed value for $0<s<t$.\footnote{In essence, this assumes that the regulator receives no new information about the future cash-flows until the year in which the cash-flow is received.} Finally, we will assume that $\mathbb{V}_{t-1 \to t}(X_t) = \mathbb{E}_{t-1}(X_t) / 1.35$. It follows that the cost of capital for the first cash-flow $X_1$ is 1.35 (as before), and for all the other cash-flows:
\begin{equation}
    \mathbb{R}_{0 \to t}(X_t) = (1.05)^{t-1}(1.35)
\end{equation}

Given these assumed values for the component costs of capital, we can use the actual cash-flow and closing asset base for the 13 electricity network distribution businesses on the east coast of Australia (known as Distribution Network Service Providers, or DNSPs) over the period 2014-2019 to determine the value of $R$ which satisfies equation \ref{eqn:c4}. The results are set out in table \ref{table:3}. As can be seen, the variation in the relative size of the cash-flows gives rise to a variation of about 400 basis points in the discount factor $R$, even though all of these DNSPs are assumed to have exactly the same underlying component costs of capital. This variation is roughly the same size as a variation in, say, the market risk premium of 0.5\%. 

\begin{table}[h!]
\singlespacing
\centering
\caption{Illustration of different values of the discount factor $R$ for real-world distribution network businesses with the same underlying component costs of capital}
\begin{tabular}{|c | c |} 
 \hline
 DNSP & Value of $R$ \Tstrut\Bstrut\\
 \hline
 Energex & 5.543\% \Tstrut\\ 
 Evoenergy & 5.581\% \\
 AusNet & 5.622\% \\
 Ergon Energy & 5.631\% \\
 CitiPower & 5.649\% \\
 United Energy & 5.688\% \\
 Powercor & 5.735\% \\
 Jemena & 5.753\% \\
 Ausgrid & 5.819\% \\
 Endeavour Energy & 5.796\% \\
 TasNetworks & 5.920\% \\
 SA Power Networks & 5.922\% \\
 Essential Energy & 5.942\%
 \Bstrut\\ 
 \hline
\end{tabular}
\label{table:3}
\end{table}

\section{The cost of capital for the debt stream}\label{sec:cod}

In standard regulatory practice the relevant cost of capital for regulatory purposes is often estimated as a weighted average of the cost of capital for the debt and equity cash-flow streams of the regulated firm.\footnote{In conventional corporate finance theory and practice, it is common to observe that the total cash-flow of the firm (denoted $X_1$, $X_2$, \ldots above) is paid out to the investors in the form of two distinct streams of payments: The payments to equity holders (in the form of dividends, plus share buybacks, less new share issues), and the payments to debt holders (in the form of interest payments, plus retirement of existing debt, less new debt issues). It is common to seek to estimate the cost of capital for the equity payment stream and the debt payment stream separately, before recombining them to find the `weighted average' cost of capital for the firm as a whole.} In addition, in standard regulatory practice the cost of capital for debt is typically estimated as the current `yield to maturity' on a corporate bond of the relevant credit rating and term. Usually the term is taken to be the same as the length of the regulatory period (e.g., five years) or longer (say, ten years). In recent years some regulators in Australia have used a cost of debt which is based on a trailing average of historic rates on corporate bonds of the relevant credit rating and term.

What can we say about the theoretically-correct approach to estimating the cost of debt in regulatory proceedings?



\subsection*{Debt preliminaries}

let's assume that we have a set of debt instruments, distinguished only in their date of maturity $t$.\footnote{All other characteristics of the debt instruments, such as the credit rating, or seniority, are assumed to be chosen so that the traded debt instruments are a perfect substitute for the debt instruments of the firm whose cost of capital we are estimating.}
The debt instrument which matures at time $t$ is assumed to make an uncertain payoff given by the random variable $I_t$ at time $t$. There are no other payments from each debt instrument (i.e., they are `zero coupon bonds').\footnote{This is without loss of generality as a regular bond paying coupons can be constructed out of a series of zero-coupon bonds.} These instruments are assumed to be actively traded in an efficient market. The current price of this debt instrument in the market is its present value $\mathbb{V}_{0 \to t}(I_t)$. It follows that the cost of capital to maturity of each of these instruments is the ratio of the expected future payoff $\mathbb{E}_0(I_t)$ to the current market price $\mathbb{V}_{0 \to t}(I_t)$:
\begin{equation}
    \mathbb{R}_{0 \to t}(I_t)=\frac{\mathbb{E}_0(I_t)}{\mathbb{V}_{0 \to t}(I_t)}
\end{equation}

At time zero the regulated firm is assumed to hold a portfolio of debt instruments, with a volume of the instrument maturing at time $t$ given by the quantity $D_{0 \to t}$ for $t=1,2, \ldots$. From the linearity of the value function, the value of this portfolio at time zero can be directly derived from the current observed price of each debt instrument:
\begin{equation}
    \sum_{t=1} \mathbb{V}_{0\to t}(D_{0 \to t}I_t) = \sum_{t=1} D_{0 \to t} \mathbb{V}_{0\to t}(I_t) 
\end{equation}

At the end of the first period, the debt instrument maturing at time $t=1$ matures, making the payoff $D_{0 \to 1} I_1$. In addition, at this time the firm is assumed to be able to make changes to its portfolio of debt instruments. We can represent this by the assumption that the firm sells its entire portfolio purchased at time zero, which has the value at time 1 of $\sum_{t=2} \mathbb{V}_{1\to t}(D_{0 \to t}I_t)$, and then purchases a new portfolio of debt instruments, given by the quantity $D_{1 \to t}$ of the instrument $I_t$, $t=2,3 , \ldots$. The cost of purchasing this new portfolio (which is also its current value) is $\sum_{t=2} \mathbb{V}_{1\to t}(D_{1 \to t}I_t)$. Ignoring transactions costs, the net payoff at time one is therefore as follows:
\begin{equation}
    X^D_1=D_{0 \to 1} I_1 + \sum_{t=2} D_{0 \to t} \mathbb{V}_{1\to t}(I_t) - \sum_{t=2} D_{1 \to t} \mathbb{V}_{1\to t}(I_t)\label{eqn:dps1}
\end{equation}
The combined payoff is therefore:
\begin{equation}
    X^D_1+\V^D_1=D_{0 \to 1} I_1 + \sum_{t=2} D_{0 \to t} \mathbb{V}_{1\to t}(I_t)\label{eqn:dps2}
\end{equation}

Importantly, the present value of the stream of debt payments at time zero depends \textit{only} on the value of the debt portfolio that is held at time zero (all future changes in the debt portfolio are of no consequence). This is demonstrated in appendix \ref{app:d}. 

\subsection*{Estimating the cost of debt in a one-year regulatory period}

Let's assume that the stream of payments to equityholders and debtholders are denoted $(X^E_1, X^E_2, \ldots)$ and $(X^D_1, X^D_2, \ldots)$, respectively. The total cash-flow of the firm is assumed to be paid out in total to equityholders and debtholders each period:
\begin{equation}
    \forall t, \; X_t=X^E_t+X^D_t\label{eqn:decomp}
\end{equation}
As above, the cash-flow stream to debtholders $X^D_t$ is determined by the portfolio of debt instruments held by the firm (that is, any payments for maturing debt instruments) plus purchases of new debt instruments less sales of old debt instruments.

Let's return now to the case of a one-year regulatory period. As we have seen, the standard regulatory practice involves estimating a single cost of capital for the firm as a whole $X_1+RAB_1$, which we labeled $\R_{0 \to 1}(X_1+RAB_1)$. 
It follows immediately from equation \ref{eqn:decomp} that the cost of capital for the combined cash-flow of the firm $\R_{0 \to 1}(X_1+RAB_1)$ can be expressed as a weighted average of the cost of capital for the equity payment stream $\R_{0 \to 1}(X^E_1+\V^E_1)$ and the cost of capital for the debt payment stream $\R_{0 \to 1}(X^D_1+\V^D_1)$, using either of the weighted average formulae in equations \ref{eqn:wa1} or \ref{eqn:wa2}.
Let's put aside the circularity problems discussed above and explore what we can say about the estimation of the cost of debt as a step toward the estimation of the cost-of-capital for the firm as a whole.

The relevant cost of debt is the cost of capital for the combined debt cash-flow  $X^D_1+\mathbb{V}^D_1$. We can express this cost of debt as the weighted average of the cost of capital of the instruments in the portfolio:
    \begin{align}
    \R_{0 \to 1}(X^D_1+\V^D_1) &= \frac{\E_0(X^D_1+\V^D_1)}{\V_{0 \to 1}(X^D_1+\V^D_1}\nonumber\\
    &=D_{0 \to 1}\frac{\E_0(I_1)}{\V^D_0}+ \sum_{t=2} D_{0 \to t} \frac{\E_0(\V_{1 \to t}(I_t))}{\V^D_0}\nonumber\\
    &=D_{0 \to 1}R_{0 \to 1}(I_1)\frac{\V_{0 \to 1}(I_1)}{\V^D_0}\nonumber\\
    &+\sum_{t=2} D_{0 \to t} \R_{0 \to 1}(\V_{1\to t}(I_t))\frac{\V_{0 \to t}(I_t)}{\V^D_0}\label{eqn:dcoc}
    \end{align}

    The weighting on each instrument in this depends on the ratio of the current price for the debt instrument in the market  $\V_{0\to t}(I_t)$ to the total value of the portfolio  $\V_0^D$. In principle both of these can be easily observed. In addition, the relevant cost of capital for the one-year debt instrument $I_1$ can, is the ratio of the expected payoff in one year $\E_0(I_1)$ to the current market price $\V_0(I_1)$:
    \begin{equation}
        \R_{0 \to 1}(I_1) = \frac{\E_0(I_1)}{\V_0(I_1)}
    \end{equation}
    In principle the expected payoff on the debt instrument can be estimated if we can estimate the probability of default and the likely recovery of funds to debtholders in the event of default. As noted above, the current market price can be easily observed. As a consequence, the relevant cost of capital for a one-year debt instrument can, in principle be estimated.

    But what about the relevant cost of capital for the longer-term debt instruments in the portfolio? For longer-term debt instruments the relevant cost of capital for each instrument in equation \ref{eqn:dcoc} is the cost of capital associated with holding the long-term debt instrument from time zero to time one: $\R_{0 \to 1}(\V_{1\to t}(I_t))$. This is equal to the expected future (time one) price of the instrument over the current price:
    \begin{equation}
        \R_{0 \to 1}(\V_{1 \to t}(I_t)) = \frac{\E_0(\V_{1 \to t}(I_t))}{\V_{0 \to 1}(\V_{1 \to t}(I_t))}=\frac{\E_0(\V_{1 \to t}(I_t))}{\V_{0 \to t}(I_t)}
    \end{equation}

Unfortunately, the expected future price of the debt instrument $\E_0(\V_{1 \to t}(I_t))$ cannot be easily observed in the market. This value is related to forecasts of future interest rates and investor tolerance of risk.

These observations are illustrated in figure \ref{fig:3}. Let's assume that the regulated firm holds a portfolio of debt instruments maturing in one, two and three years, labelled $I_1$, $I_2$, and $I_3$. Each of these instruments has a `face value' of, say, \$1000. The current price of these three instruments is \$900, \$800, and \$700, say. The regulator can in principle estimate the expected payout on the instrument maturing in one year, $I_1$. Although it has a face value of \$1000, the actual expected payout will be somewhat less, reflecting the probability of default and the recovery expected in the event of default. Let's suppose that the expected future payout is, say, \$950. In this case the cost of capital for this instrument is $\$950/\$900-1 = 5.56\%$. But in the case of the longer-term instruments $I_2$ and $I_3$ it is not easy to estimate the price of these instruments at time one. As a result estimating the cost of capital between time 0 and time 1 for these instruments is not straightforward.

\begin{figure}[ht]
\singlespacing
    \centering
    \caption{It is not, in general, possible to estimate the cost of capital for a portfolio of debt instruments using currently observed market data}
    \begin{tikzpicture}
    \draw [->] (0,2.5) to (2,2.5);
    \draw [->] (0,1) to (4,1);
    \draw [->] (0,-0.5) to (6,-0.5);
    \draw [->,dashed, bend left] (0,1) to (2,1);
    \draw [->,dashed, bend left] (0,-0.5) to (2,-0.5);
    \draw [->,dashed, bend left] (2,1) to (4,1);
    \draw [->,dashed, bend left] (2,-0.5) to (6,-0.5);
    \node [below] at (-1,-0.8) {Time};
    \node [below] (A) at (0,-0.8) {0};
    \node [below] (B) at (2,-0.8) {1};
    \node [below] (C) at (4,-0.8) {2};
    \node [below] (C1) at (6,-0.8) {3};
    \node [above, text width =1.5 cm, align =center] (D) at (0,3.3) {Current price};
    \node [above, text width =1.5 cm, align =center] (E) at (2,3.3) {Expected time 1 payoff};
    \node [left] (I1) at (-0.5,3) {$I_1$};
    \node [left] (I2) at (-0.5,1.7) {$I_2$};
    \node [left] (I3) at (-0.5,0.4) {$I_3$};
    \node  (F1) at (0,3) {\$900};
    \node  (G1) at (0,1.7) {\$800};
    \node  (H1) at (0,0.4) {\$700};
    \node (L1) at (2,3) {$\E_0(I_1)=\$950$};
    \node (L2) at (2,1.7) {$\E_0(\V_{1 \to 2}(I_2))=?$};
    \node (L3) at (2,0.4) {$\E_0(\V_{1 \to 3}(I_3))=?$};
    \node [draw, text width=3 cm, align = center] (box1) at (-3,1.5) {At time 0 the current prices for each debt instrument can be observed in the market.};
    \node [draw, text width=3 cm, align = center] (box2) at (5.5,4) {The future payoff of the debt instrument maturing at time 1 can be estimated.};
    \node [draw, text width=3.3 cm, align = center] (box3) at (6.2,1.5) {The time 1 prices of the debt instruments maturing in the future cannot be easily estimated.};
    \draw [->] (box1) to (I1);
    \draw [->] (box1) to (I2);
    \draw [->] (box1) to (I3);
    \draw [->] (box2) to (L1);
    \draw [->] (box3) to (L2);
    \draw [->] (box3) to (L3);
    \end{tikzpicture}
    \label{fig:3}
\end{figure}
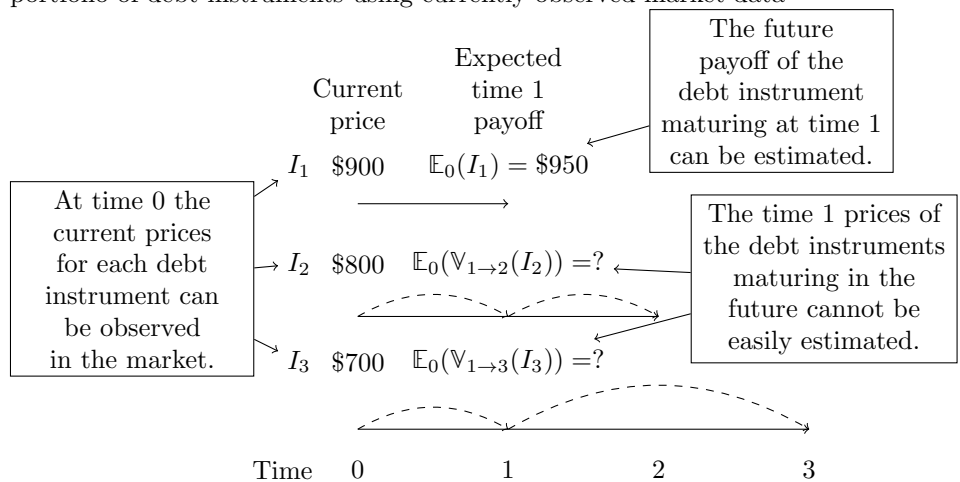

We are now in a position to answer one of the questions asked at the outset: Can we estimate the cost of debt for a regulated firm by observing the appropriate yield-to-maturity on a corporate bond of the right term? If so, what is the right term?

This analysis provides no support for the assertion that, in standard regulatory practice, we can estimate the cost of debt for a regulated firm by observing the yield-to-maturity on a corporate bond of the right term, for the following reasons:
\begin{itemize}
    \item In standard, historic regulatory practice (as summarised in equation \ref{eqn:bbm9} or equations \ref{eqn:bbm10} and \ref{eqn:bbm11}), the regulator is interested in estimating the cost of capital of the combined cash-flow for the firm as a whole. If we seek to estimate this cost of capital as a weighted average of the cost of equity and the cost of debt the relevant cost of capital for debt is the cost of capital for the debt payment stream of the firm. $\R_{0 \to 1}(X^D_1+\V^D_1)$. This depends on the \textit{portfolio of debt instruments} held by the regulated firm (and not just the cost of capital for a single bond).

    \item In general, the debt portfolio held by the regulated firm will include instruments with a range of terms. Although it is relatively easily to estimate the cost of capital for a debt instrument maturing in one year, it is not straightforward to estimate the one-year cost of capital for debt instruments maturing in future years.

    \item The current yield-to-maturity of a debt instrument is the ratio of the `face value' to the current price. But the relevant cost of capital $\R_{0 \to t}(I_t)$ is the ratio of the expected future payout to the current price:
    \begin{equation}
        \R_{0 \to t}(I_t)=\frac{\E_0(I_t)}{\V_{0 \to t}(I_t)}
    \end{equation}
    For any debt instrument other than a risk-free instrument, the expected future payout $\E_0(I_t)$ is less than the face value due to the probability of default. Therefore, the yield-to-maturity on a bond is an over-estimate of the cost of capital for that bond.
\end{itemize}

If the cost of capital for the debt portfolio of the regulated firm cannot be easily observed (as appears to be normally assumed) it follows that there appears to be little value in separating the cost of capital for the firm as a whole into a weighted average of the cost of equity and the cost of debt.

This leaves two possibilities: We could estimate the cost of capital for the firm as a whole (as before) -- or, more strictly, the component costs of capital for the firm as a whole, as set out in equation \ref{eqn:b3}. Alternatively, we could implement a version of the regulatory process which only requires estimation of the (components of the) costs of equity.

To see how this might be achieved, let's assume that the cash-flow stream to debt can be treated as exogenous -- determined outside the regulatory process. The regulatory process then determines the equity cash-flow $X^E_1$ and the equity asset base $RAB^E_1$. The required version of the Building Block Model equations can be derived as follows. First, equation \ref{eqb:bbm} can be expanded as follows:
\begin{align}
    RAB_0 &= \V_{0}(X_1+RAB_1) = \V_{0}(X^E_1+RAB^E_1 + X^D_1+\V^D_1)\nonumber\\
    &= \frac{\E_{0}(X^E_1)}{\R_{0 \to 1}(X^E_1)}+\frac{\E_{0}(RAB^E_1)}{\R_{0 \to 1}(RAB^E_1)} + \V_{0 \to 1}(X^D_1+\V^D_1)
\end{align}
It follows that the allowed cash-flow stream to equity of the regulated firm should be determined as follows:
\begin{align}
    \E_0(X^E_1) &= RAB_0 \times \R_{0 \to 1}(X^E_1) 
    - \E_{0}(RAB^E_1) \frac{\R_{0 \to 1}(X^E_1)}{\R_{0 \to 1}(RAB^E_1}\nonumber\\
    &+ \V_{0 \to 1}(X^D_1+\V^D_1) \R_{0 \to 1}(X^E_1)\nonumber\\
    &=RAB_0 \times \R_{0 \to 1}(X^E_1)  \nonumber\\
    &- \left(\E_{0}(RAB^E_1)- \R_{0 \to 1}(RAB^E_1) \sum_{t=1} D_{0 \to t} \V_{0 \to t}(I_t)\right)\frac{\R_{0 \to 1}(X^E_1)}{\R_{0 \to 1}(RAB^E_1)}\label{eqn:f1}
\end{align}

This is a further variation on the Building Block Model, extending equation \ref{eqn:b3} by replacing the expected closing asset base $\E_{0}(RAB^E_1)$ with an expression that subtracts the current value of the debt portfolio: $\E_{0}(RAB^E_1)- \R_{0 \to 1}(RAB^E_1) \sum_{t=1} D_{0 \to t} \V_{0 \to t}(I_t)$. This last term ($\sum_{t=1} D_{0 \to t} \V_{0 \to t}(I_t)$) can be directly observed from market data.



We have seen that, in the case of a single-year regulatory period, if we use a single cost of capital for the combined cash-flow of the firm it is possible to write that cost of capital as a weighted average of the (combined) cost of capital for debt and the (combined) cost of capital for equity. If we seek to implement a version of the regulatory process in which we estimate separate costs of capital for the components (as in equation \ref{eqn:b3}) then again we can express the costs of capital for those components as a weighted average of the corresponding component cost of capital for debt and equity separately.

What about the case of a multi-year regulatory period? Let's suppose that we seek to estimate a single cost of capital for the entire regulatory period (as in equation \ref{eqn:c2}. Can we express this cost of capital parameter $R$ as a weighted average of the corresponding parameter for equity and for debt?

The answer is no. Let's suppose that the parameter $R$ satisfies equation \ref{eqn:c4} and the corresponding parameter $R^E$ and $R^D$ satisfies the corresponding equation for debt and equity. In the case where $T=2$ this yields the following:
\begin{align}
    &\frac{\E_0(X_1)}{R}+\frac{\E_0(X_2+RAB_3)}{R^2}\nonumber\\
    &=\frac{\E_0(X^E_1)}{R^E}+\frac{\E_0(X^E_2+RAB^E_2)}{{R^E}^2}+\frac{\E_0(X^D_1)}{R^D}+\frac{\E_0(X^D_2+\V^D_2)}{{R^D}^2}
\end{align}
In this case the parameter $R$ cannot be expressed as a weighted average of $R^D$ and $R^E$. This calls into question the common use of a weighted-average cost of capital in the context of a multi-year regulatory period.

\section{Conclusion}\label{sec:con}

Cost of capital issues are amongst the most controversial in regulatory practice. In my view, these debates have been made more clouded and confused by the lack of a strong theoretical foundation. In the absence of a clear, strong, theoretical foundation, regulators and courts are not in a position to make reasoned, rational choices between one approach or methodology and another. This tends to prolong and perpetuate disputes.

Cost of capital for regulatory purposes tends to draw heavily and uncritically on the corporate finance literature. But the corporate finance context tends to be quite different from the regulatory context, with different objectives and assumptions. Approaches which are commonplace in the corporate finance literature (such as the estimation of the cost of capital as a weighted average of the cost of equity and the cost of debt) do not necessarily carry over to the regulatory context.

This article seeks clarify and formalise the theory of cost of capital for regulatory purposes. Because the cost of capital for a regulated firm depends strongly on the precise formulation of the regulatory approach, this article has sought to be clear about the standard regulatory approach, and various possible variations. A starting point of this analysis is the assumption that (putting aside incentive concerns), a central objective of all cost-based regulation is the achievement of NPV=0. Errors in the estimation of the cost of capital potentially undermine this objective.

There are different ways of setting an allowed revenue stream to achieve an overall NPV of zero. Those different approaches to regulating will, in general, require a different corresponding cost of capital (or potentially, multiple different costs of capital). If there was one cost of capital that was materially easier to accurately estimate than others, we might favour the corresponding approach to regulation. But this is not obviously the case.

This analysis has suggested the following problems with, and potential improvements to standard regulatory practice:
\begin{itemize}
    \item In standard regulatory practice (in the one period case, as summarised in equation \ref{eqn:bbm9} or equations \ref{eqn:bbm10} and \ref{eqn:bbm11}), the regulatory process depends on estimates of a single cost of capital for the regulated firm as a whole, which we have referre to as the combined cost of capital, denoted $\R_{0 \to 1}(X_1+RAB_1)$. This single cost of capital depends on the level of both the one-period cash-flow $\E_0(X_1)$ and the closing asset base $\E_0(RAB_1)$. But this cost of capital is itself an input to the determination of the regulated cash-flow allowance, giving rise to a problem of circularity. At a minimum this muddies the problem of estimation of the cost of capital under this regulatory process.
    
    The regulatory process could be made clearer and cleaner by changing the regulatory process so that it makes use of a separate cost of capital for $X_1$ and for $RAB_1$ individually (which we have referred to as the component costs of capital). In the context where the closing asset base $RAB_1$ is chosen by the regulator, the cost of capital for this component is just the risk-free rate. But this still leaves the problem of estimating the cost of capital for $X_1$. If this could be done effectively it would result in more effective achievement of the fundamental objective of NPV=0.

    \item In the context of a multi-year regulatory period, the standard regulatory approach is to use a single cost of capital for the entire regulatory period. This single cost of capital parameter is the solution to an `internal rate of return' calculation which depends on the level of the cash-flow allowance $E_0(X_t)$ in each year of the regulatory period and the level of the closing asset base $\E_0(RAB_T)$. This cost of capital parameter suffers from the same problem of circularity mentioned above.
    
    Many regulators have argued that the appropriate term of this cost of capital is equal to the length of the regulatory period. This practice is not supported in the theory set out here. The relevant single cost of capital parameter is a complicated mix of costs of capital of various terms, shorter than and equal to the length of the regulatory period.
    
    The regulatory process could be made clear and cleaner by changing the regulatory process to use a different cost of capital for each component of the cash-flow $X_1, X_2, \ldots$ over the period. There are different ways of setting the allowed revenue over the regulatory period, but, if the revenues are set annually, the relevant cost of capital for each cash-flow individually is a one-year rate (either the forward rate, or the out-turn rate, depending on the approach used). But, in any case, the relevant term of the cost of capital is not equal to the length of the regulatory period.

    \item The standard regulatory approach estimates the cost of capital as the weighted average of the cost of equity and the cost of debt. The cost of debt is estimated in different ways, but one typical way is to estimate the cost of debt as the yield-to-maturity on a corporate bond of a particular credit rating and term. The analysis set out here does not support that practice. The relevant cost of debt is the cost of capital for the debt portfolio of the regulated firm. Although the cost of capital for debt instruments maturing in one year may (in theory) be estimated drawing on knowledge of the current price of those instruments in the market, the cost of capital for debt instruments maturing in future years cannot be observed using current market data. As a result it is not, in general, possible to easily estimate the cost of capital for the debt portfolio of the regulated firm.
    
    This calls into question the value of separately estimating the cost of equity and the cost of debt and combining them to form an estimate of the cost of capital for the firm as a whole. If neither the cost of equity nor the cost of debt can be easily estimated, it is not clear that this approach offers an improvement over simply estimating the cost of capital for the firm as a whole.

    It is possible to construct a regulatory process in which there is no need to estimate a cost of debt (only a cost of equity would be required). This regulatory process would isolate and focus on the cash-flow stream to the equity of the regulated firm (summarised in equation \ref{eqn:f1}). However, it is not clear that this approach offers an improvement over simply estimating the cost of capital for the firm as a whole.

    \item In the context of a multi-year regulatory period, the single cost of capital parameter cannot be expressed as a weighted average of a similar parameter for equity and for debt. Again, this calls into question the value of separately estimating the cost of equity and the cost of debt and combining them to form an estimate of the cost of capital for the firm as a whole.

\end{itemize}

There is, at present, no known mechanism for achieving the fundamental objective of regulation (NPV=0) without estimating some form of the cost of capital. We cannot know how to improve our estimates of the cost of capital for regulatory purposes without a clear understanding of the underlying theory. In my view, clarifying and articulating that theory -- as set out in this article -- is at least a start in placing regulatory practice on a sound footing going forward.

\begin{appendices}
\section{Derivation of CAPM}\label{app:a}

This appendix derives the CAPM using the value functional and assumptions on the preferences of the representative investor.

Let's assume that the representative investor has mean-variance preferences. In other words, let's assume that the utility of the representative investor from a certain income $Y$ today (at time 0) and an uncertain income $X$ at time 1, is given by the following:
\begin{equation}
    U(X)=Y + \delta ( E[X]-\alpha Var[X])
\end{equation}
Here $\alpha$ is a parameter which reflects the degree of risk aversion of the investor, and $\delta$ is a parameter which reflects the rate of time preference between time 0 and time 1.

Let's suppose the set of all assets in the economy (the so-called `market portfolio') is represented in the payoff $M$ at time 1. This is the set of all assets which must be held by the investors in equilibrium.

Now consider a small change to the equilibrium that involves the addition of a small amount $\epsilon$ of an asset $X$. In equilibrium this asset must be priced in a way such that the purchase of a small amount does not change the utility of the representative investor holding the market portfolio. The utility from purchasing and the portfolio $M+ \epsilon X$ at time 0 and holding it to time 1 is as follows:
\begin{equation}
    U(M+\epsilon X) = \delta ( \E[M+\epsilon X]-\alpha Var[M+\epsilon X]) - \V(M+\epsilon X)
\end{equation}
The first order condition with respect to $\epsilon$ (setting $\epsilon=0$) is as follows:
\begin{equation}
    \V(X) = \delta \E[X] - 2 \alpha \delta Cov(M,X)
\end{equation}
The two parameters $\alpha$ and $\delta$ can be determined using the results: (a) when $X$ has a certain payoff, $\V(X)=\Delta_F \E(X)$, where $\Delta_F=RF^{-1}$ is the inverse of the risk-free cost of capital and (b) in the case where $X$ is a share of the market portfolio, $\V(M)= \Delta_M \E(M)$, where $\Delta_M=\R_{0 \to 1}(M_1)^{-1}$ is the inverse of the cost of capital for the market portfolio. This yields:
\begin{equation}
    \Delta_X= \Delta_F - (\Delta_F-\Delta_M) \beta(X)\label{eqn:capm1}
\end{equation}
Here $\Delta_X=\R_{0 \to 1}(X_1)^{-1}$ is the inverse of the cost of capital for the cash-flow $X_1$. Equation \ref{eqn:capm1} is the version of the CAPM in this context.

\section{The Fundamental Theorem}\label{app:b}

This appendix sets out a proof of the Fundamental Theorem of Regulation.

\begin{theorem}

Suppose that, at the start of a regulatory period, in year $t$, the regulator observes $RAB_t$ and chooses $X_{t+1}, X_{t+2}, \ldots, X_{t+T}$ and $RAB_{t+T}$, and the parameters $A_{t \to t+1}, A_{t \to t+2}, \ldots A_{t \to t+T}$ and $B_{t \to t+T}$ to satisfy the following two conditions:
\begin{align}
     \mathbb{V}_t(X_{t+1}, &X_{t+2}, \ldots, X_{t+T}+RAB_{t+T})  \nonumber\\
     &=\frac{\mathbb{E}_t(X_{t+1})}{A_{t \to t+1}} + \frac{\mathbb{E}_t(X_{t+1})}{A_{t \to t+1}} +\ldots + \frac{\mathbb{E}_t(X_{t+T})}{A_{t \to t+T}} +\frac{\mathbb{E}_t(RAB_{t+T})}{B_{t \to t+T}}\nonumber\\
     &= RAB_t\label{eqn:beren}
\end{align}
And, in addition, at some point in the future $s$ (e.g., at the end of the life of the firm), the regulator ensures that $RAB_s=\mathbb{V}_s(X_{s+1}, X_{s+2}, \ldots)$. Then, at time $t$, the asset base is equal to the present value of the future stream of cash-flows:
\begin{align}
     \mathbb{V}_t(X_{t+1}, X_{t+2}, \ldots) &= RAB_t
\end{align}
It follows that the firm achieves NPV=0.
\end{theorem}
\begin{proof}
By backward induction.
\end{proof}

\section{Five year regulatory period}\label{app:c}

Let's suppose that the regulator follows the following practice: At the start of each year of the regulatory period, that is at time $t-1$, ($t=1, \ldots T$), the value of the parameters $A_{t-1 \to t}$ and $B_{t-1 \to t}$ are determined and the regulated cash-flow allowance $X_t$ and the closing asset base $RAB_t$ is set to satisfy the following equation:
\begin{equation}
    \E_{t-1}(X_t)= A_{t-1 \to t} RAB_{t-1} - \E_{t-1}(RAB_t) \frac{A_{t-1 \to t}}{B_{t-1 \to t}}
\end{equation}
This can, of course, be re-written as the requirement that, at the start of each year of the regulatory period, the regulated cash-flow allowance $X_t$ and the closing asset base $RAB_t$ is set to satisfy the following:
\begin{equation}
    RAB_{t-1} = \frac{\E_{t-1}(X_t)}{A_{t-1 \to t}} + \frac{RAB_{t-1}}{B_{t-1 \to t}}
\end{equation}
Expanding this equation over the five-year (say) regulatory period starting at time $t=1$ we have the following:
\begin{align}
    RAB_0 &= \frac{\E_0(X_1)}{A_{0 \to 1}}+ \frac{\E_0}{B_{0 \to 1}}\left[\frac{\E_1(X_2)}{A_{1 \to 2}}\right]\nonumber\\
    &+ \frac{\E_0}{B_{0 \to 1}}\left[\frac{\E_1}{B_{1 \to 2}}\left[\frac{\E_2(X_3)}{A_{2 \to 3}}\right]\right]\nonumber\\
    &+ \frac{\E_0}{B_{0 \to 1}}\left[\frac{\E_1}{B_{1 \to 2}}\left[\frac{\E_2}{B_{2 \to 3}}\left[\frac{\E_3(X_4)}{A_{3 \to 4}}\right]\right]\right]\nonumber\\
    &+ \frac{\E_0}{B_{0 \to 1}}\left[\frac{\E_1}{B_{1 \to 2}}\left[\frac{\E_2}{B_{2 \to 3}}\left[\frac{\E_3}{B_{3 \to 4}}\left[\frac{\E_4(X_5)}{A_{4 \to 5}}\right]\right]\right]\right]\nonumber\\
    &+ \frac{\E_0}{B_{0 \to 1}}\left[\frac{\E_1}{B_{1 \to 2}}\left[\frac{\E_2}{B_{2 \to 3}}\left[\frac{\E_3}{B_{3 \to 4}}\left[\frac{\E_4(RAB_5)}{B_{4 \to 5}}\right]\right]\right]\right] \label{eqn:bigexp}
\end{align}

The question is how to choose the values of the cost-of-capital parameters $A_{0 \to 1}, A_{1 \to 2}, \ldots, A_{t-1 \to T}$ (and similarly for $B$) in order to satisfy equation \ref{eqn:beren}.

To begin, we will choose $A_{t-1 \to t}=\R_{{t-1} \to t}(X_t)$, $t=1, \ldots T$. The first term in equation \ref{eqn:bigexp} then becomes:
\begin{equation}
    \frac{\E_0(X_1)}{A_{0 \to 1}}=\frac{\E_0(X_1)}{\R_{0 \to 1}(X_1)}=\V_{0 \to 1}(X_1)
\end{equation}
As required. For the second term we choose:
\begin{equation}
B_{0 \to 1}=\R_{0 \to 1}(V_{1\to 2}(X_2))     
\end{equation}
Then the second term becomes:
\begin{equation}
    \frac{\E_0}{B_{0 \to 1}}\left[\frac{\E_1(X_2)}{A_{1 \to 2}}\right]=\frac{\E_0}{B_{0 \to 1}}\left[\V_{1 \to 2}(X_2)\right]=\V_{0 \to 2}(X_2)
\end{equation}
As required. For the third term we choose:
\begin{equation}
 B_{1 \to 2}= \frac{\R_{0 \to 1}(\V_{1\to 3}(X_3)) \R_{1 \to 2}(\V_{2\to 3}(X_3))}{ B_{0 \to 1}}
\end{equation}
The third term then becomes:
\begin{align}
    &\frac{\E_0}{B_{0 \to 1}}\left[\frac{\E_1}{B_{1 \to 2}}\left[\V_{2 \to 3}(X_3)\right]\right]\nonumber\\
    & =
    \frac{\E_0}{B_{0 \to 1}}\left[ \frac{B_{0 \to 1}}{\R_{0 \to 1}(\V_{1\to 3}(X_3))} \frac{\E_1(\V_{2 \to 3}(X_3))}{\R_{1 \to 2}(\V_{2\to 3}(X_3))}\right]\nonumber\\
    &=\frac{\E_0}{\R_{0 \to 1}(\V_{1\to 3}(X_3))}\left[  \V_{1\to 2}(\V_{2 \to 3}(X_3))\right]\nonumber\\
    &=\V_{0\to 3}(X_3)
\end{align}
For the fourth term we choose:
\begin{equation}
    B_{2 \to 3}= \frac{\R_{0 \to 1}(\V_{1\to 4}(X_4))\R_{1 \to 2}(\V_{2\to 4}(X_4)) \R_{2 \to 3}(\V_{3\to 4}(X_4))}{B_{0 \to 1} B_{1 \to 2} }
\end{equation}
(We will omit the algebra in this case for brevity). Similarly, for the fifth term we choose:
\begin{equation}
    B_{3 \to 4}= \frac{\R_{0 \to 1}(\V_{1\to 5}(X_5))\R_{1 \to 2}(\V_{2\to 5}(X_5))\R_{2 \to 3}(\V_{3\to 5}(X_5)) \R_{3 \to 4}(\V_{4\to 5}(X_5))}{B_{0 \to 1} B_{1 \to 2} B_{2 \to 3}}
\end{equation}
The fifth term then becomes:
\begin{align}
&\frac{\E_0}{B_{0 \to 1}}\left[\frac{\E_1}{B_{1 \to 2}}\left[\frac{\E_2}{B_{2 \to 3}}\left[\frac{\E_3}{B_{3 \to 4}}\left[\V_{4 \to 5}(X_5)\right]\right]\right]\right]\nonumber\\
    &=\frac{\E_0}{B_{0 \to 1}}\left[\frac{\E_1}{B_{1 \to 2}}\left[\frac{B_{0 \to 1} B_{1 \to 2}}{\R_{0 \to 1}(\V_{1\to 5}(X_5))\R_{1 \to 2}(\V_{2\to 5}(X_5))}\left[\frac{\E_2[\V_{3\to 5}(X_5)]}{\R_{2 \to 3}(\V_{3\to 5}(X_5))}\right]\right]\right]\nonumber\\
    &=\frac{\E_0}{B_{0 \to 1}}\left[\frac{B_{0 \to 1} }{\R_{0 \to 1}(\V_{1\to 5}(X_5))}\left[\frac{\E_1[\V_{2\to 5}(X_5)]}{\R_{1 \to 2}(\V_{2\to 5}(X_5))}\right]\right]\nonumber\\
    &=\frac{\E_0[\V_{1\to 5}(X_5)]}{\R_{0 \to 1}(\V_{1\to 5}(X_5))}\nonumber\\
    &=\V_{0\to 5}(X_5)
\end{align}
The choice of $B_{4 \to 5}$ follows similarly from the algebra.

\section{Changes in the debt portfolio}\label{app:d}

The total payoff at time one from the debt portfolio $X^D_1+\V^D_1$ is independent of any future changes in the debt portfolio. This can be easily demonstrated in a simple case:

Let's suppose that, at time zero, the firm holds amount $D_{0 \to 1}$ of instrument $I_1$ maturing at time 1, amount $D_{0 \to 2}$ of $I_2$ maturing at time 2, and amount $D_{0 \to 3}$ of $I_3$ maturing at time 3. The value of this portfolio at time zero is:
\begin{equation}
    \V^D_0=D_{0 \to 1} \V_{0 \to 1}(I_1) + D_{0 \to 2} \V_{0 \to 2}(I_2) + D_{0 \to 3} \V_{0 \to 3}(I_3)
\end{equation}
At time $t=1$ the debt instrument $I_1$ matures paying the amount $D_{0 \to 1} I_1$. In addition, the firm can adjust its portfolio of the other debt instruments. It can sell its existing portfolio $D_{0 \to 2}$ of $I_2$ and $D_{0 \to 3}$ of $I_3$ and purchase the amount $D_{1 \to 2}$ of $I_2$ and amount $D_{1 \to 3}$ of $I_3$. The net cash-flow at time $t=1$ is therefore:
\begin{equation}
    X^D_1=D_{0 \to 1} I_1 + (D_{0 \to 2} - D_{1 \to 2}) \V_{1 \to 2}(I_2) + (D_{0 \to 3}- D_{1 \to 3}) \V_{1 \to 3}(I_3)    
\end{equation}
Assuming there are no further changes in the portfolio, the firm receives a payout in the amount of $D_{1\to 2}$ of $I_2$ at time $t=2$ and amount $D_{1\to 3}$ of $I_3$ at time $t=3$. This payment stream has value at time $t=1$ of $D_{1\to 2}\V_{1\to 2}(I_2)+D_{1 \to 3} \V_{1 \to 3}(I_3)$. The total value of the debt payment stream at time zero is therefore independent of any subsequent changes in the portfolio:
\begin{align}
    \V_{0 \to 1}(X^D_1+\V^D_1) & = D_{0 \to 1} \V_{0 \to 1}(I_1)\nonumber\\
    &+ (D_{0 \to 2}-D_{1 \to 2})\V_{0 \to 1}( \V_{1 \to 2}(I_2))\nonumber\\
    &+ (D_{0 \to 3}-D_{1 \to 3}) \V_{0 \to 1}(\V_{1 \to 3}(I_3))\nonumber\\
    &+ D_{1 \to 2} \V_{0 \to 1}(\V_{1\to 2}(I_2)) + D_{1 \to 3}  \V_{0 \to 1}(\V_{1 \to 3}(I_3))\nonumber\\
    &= D_{0 \to 1} \V_{0 \to 1}(I_1) + D_{0 \to 2} \V_{0 \to 2}(I_2) + D_{0 \to 3} \V_{0 \to 3}(I_3)\nonumber\\
    &=\V^D_0
\end{align}

The general proof is as follows:
\begin{align}
    \V_{0 \to 1}(X^D_1+\V^D_1) &=\V_{0 \to 1}(X^D_1+\sum_{t=2} \mathbb{V}_{1 \to t}(D_{1 \to t} I_t) )\nonumber\\
    &=  D_{0 \to 1} \mathbb{V}_{0 \to 1}( I_1) + \sum_{t=2} D_{0 \to t} \V_{0 \to 1}(\V_{1\to t}(I_t)) )\nonumber\\
    &=\sum_{t=1} D_{0 \to t} \mathbb{V}_{0 \to t}(  I_t) = \V^D_0
\end{align}
Although the present value of the combined cash-flow $X^D_1+\mathbb{V}^D_1$ is independent of the future changes in the debt portfolio, this is not true for the cash-flows $X^D_1$ and $\mathbb{V}^D_1$ separately -- these depend on the details of the changes in the debt portfolio that occur at time 1.

\end{appendices}

\bibliography{references}

\end{document}